\documentclass[journal]{IEEEtran}
\usepackage{amsmath}
\usepackage{amssymb}
\usepackage{amsfonts}
\usepackage{graphicx}
\usepackage{epsfig}
\usepackage{subfigure}
\usepackage{psfrag}
\usepackage{color}

\begin{document}
\title{Optimal Save-Then-Transmit Protocol for Energy Harvesting Wireless Transmitters}

\author{Shixin Luo, Rui Zhang,~\IEEEmembership{Member,~IEEE}, and Teng Joon Lim,~\IEEEmembership{Senior Member,~IEEE}

\thanks{S. Luo and T. J. Lim are with the Department of
Electrical and Computer Engineering, National University of
Singapore (e-mail:\{shixin.luo, eleltj\}@nus.edu.sg).}
\thanks{R. Zhang is with the Department of Electrical and Computer Engineering, National
University of Singapore (e-mail: elezhang@nus.edu.sg). He is also
with the Institute for Infocomm Research, A*STAR, Singapore.}}
\maketitle

\begin{abstract}\label{sec:abstract}
In this paper, the design of a wireless communication device relying
exclusively on energy harvesting is considered. Due to the inability
of rechargeable energy sources to charge and discharge at the same
time, a constraint we term the \emph{energy half-duplex constraint},
two rechargeable energy storage devices (ESDs) are assumed so that
at any given time, there is always one ESD being recharged. The
energy harvesting rate is assumed to be a random variable that is
constant over the time interval of interest. A
\emph{save-then-transmit} (ST) protocol is introduced, in which a
fraction of time $\rho$ (dubbed the save-ratio) is devoted
exclusively to energy harvesting, with the remaining fraction
$1-\rho$ used for data transmission. The ratio of the energy
obtainable from an ESD to the energy harvested is termed the
\emph{energy storage efficiency}, $\eta$. We address the practical
case of the secondary ESD being a battery with $\eta < 1$, and the
main ESD being a super-capacitor with $\eta = 1$. Important properties of the optimal
save-ratio that minimizes outage probability are derived, from which useful design guidelines are drawn. In addition, we compare the
outage performance of random power supply to that of constant power
supply over the Rayleigh fading channel. The diversity order with
random power is shown to be the same as that of constant power, but
the performance gap can be large. Finally, we extend the proposed
ST protocol to wireless networks with multiple transmitters. It is shown that the system-level outage
performance is critically dependent on the number of transmitters and the optimal save-ratio for
single-channel outage minimization.
\end{abstract}


\IEEEpeerreviewmaketitle
\setlength{\baselineskip}{1.0\baselineskip}
\newtheorem{definition}{\underline{Definition}}[section]
\newtheorem{fact}{Fact}
\newtheorem{assumption}{Assumption}
\newtheorem{theorem}{\underline{Theorem}}[section]
\newtheorem{lemma}{\underline{Lemma}}[section]
\newtheorem{corollary}{\underline{Corollary}}[section]
\newtheorem{proposition}{\underline{Proposition}}[section]
\newtheorem{example}{\underline{Example}}[section]
\newtheorem{remark}{\underline{Remark}}[section]
\newtheorem{algorithm}{\underline{Algorithm}}[section]
\newcommand{\mv}[1]{\mbox{\boldmath{$ #1 $}}}

\section{Introduction}\label{sec:introduction}
The operation of communication networks powered either largely or
exclusively by renewable sources has become increasingly attractive,
both due to the increased desire to reduce energy consumption in
human activities at large, and due to necessity brought about by the
concept of networking heterogeneous devices ranging from medical
sensors on/in the human body to environment sensors in the
wilderness \cite{Akyidiz2002,Sina2001}. Sensor nodes are powered by
batteries that often cannot be replaced because of the
inaccessibility of the devices. Therefore, once the battery of a
sensor node is exhausted, the node dies. Thus the potentially
maintenance-free and virtually perpetual operation offered by energy
harvesting, whereby energy is extracted from the environment, is
appealing.

The availability of an inexhaustible but unreliable energy source
changes a system designer's options considerably, compared to the
conventional cases of an inexhaustible reliable energy source
(powered by the grid), and an exhaustible reliable energy source
(powered by batteries). There has been recent research on
understanding data packet scheduling with an energy harvesting
transmitter that has a rechargeable battery, most of which employed
a deterministic energy harvesting model. In \cite{YangUlukus101},
the transmission time for a given amount of data was minimized
through power control based on known energy arrivals over all time.
Structural properties of the optimum solution were then used to
establish a fast search algorithm. This work has been extended to
battery limited cases in \cite{Yener10}, battery imperfections in
\cite{Deviller11, Yener12}, and the Gaussian relay channel in \cite{Huang11}.
Energy harvesting with channel fading has been investigated in
\cite{UlukusYener11} and \cite{HoZhang11}, wherein a water-filling
energy allocation solution where the so-called water levels follow a
staircase function was proved to be optimal.

In scenarios where multiple energy harvesting wireless devices
interact with each other, the design needs to adopt a system-level
approach \cite{Koksal2010, iannello2011, Sharma08}. In \cite{iannello2011},
the medium access control (MAC) protocols for single-hop wireless
sensor networks, operated by energy harvesting capable devices, were
designed and analyzed. In \cite{Sharma08}, $N$ energy harvesting nodes with independent data and energy queues were considered, and the queue stability was analyzed under different MAC protocols. An information theoretic analysis of energy
harvesting communication systems has been provided in
\cite{OzelUlukus10, Rajesh2011}. In \cite{OzelUlukus10}, the authors
proved that the capacity of the AWGN channel with stochastic energy
arrivals is equal to the capacity with an average power constraint
equal to the average recharge rate. This work has been extended in
\cite{Rajesh2011} to the fading Gaussian channels with perfect/no
channel state information at the transmitter.

Due to the theoretical intractability of online power scheduling
under the energy causality constraint (the cumulative energy
consumed is not allowed to exceed the cumulative energy harvested at
every point in time), most current research is focused on an offline
strategy with deterministic channel and energy state information,
which is not practical and can only provide an upper bound on
system performance. An earlier line of research considers the
problem of energy management, with only causal energy state
information, in communications satellites \cite{Alvin03}, which
formulated the problem of maximizing a reward that is linear in
the energy as a dynamic programming problem. In \cite{sharma10TWC},
energy management policies which stabilize the data queue have been
proposed for single-user communication under linear energy-rate
approximations.

In this paper, we focus our study on the design of practical circuit model and transmission protocol for energy harvesting wireless transmitters. To be more specific, we consider a wireless system with one transmitter
and one receiver, with the transmitter using a {\it
save-then-transmit} (ST) protocol (see Fig. \ref{fig:Systemmodel})
to deliver $Q$ bits within $T$ seconds, the duration of a
transmission frame. Because rechargeable energy storage devices
(ESDs) cannot both charge and discharge simultaneously (the {\em energy half-duplex constraint}), an
energy harvesting transmitter needs two ESDs, which we call the main
ESD (MESD) and secondary ESD (SESD). The transmitter draws power
from the MESD for data transmission, over which time the SESD is
connected to the energy source and charges up. At the end of
transmission for a frame, the SESD transfers its stored energy to
the MESD. A fraction $\rho$ (called the save-ratio) of every frame
interval is used exclusively for energy harvesting by the
MESD.\footnote{Note that the energy source can be connected only to
either the SESD or the MESD, but not both.} The {\it energy storage
efficiency}, denoted by $\eta$, of each ESD may not be 100 percent,
and a fixed amount of power $P_c$ is assumed to be consumed by the
transmitter hardware whenever it is powered up. The frame interval
$T$ is assumed to be small relative to the time constant of changes
in the ESD charging rate (or energy arrival rate). The energy
arrival rate is therefore modeled as a random variable $X$ in
Joules/second, which is assumed to be constant over a frame.

\begin{figure}
\centering
\epsfxsize=0.7\linewidth
    \includegraphics[width=8cm]{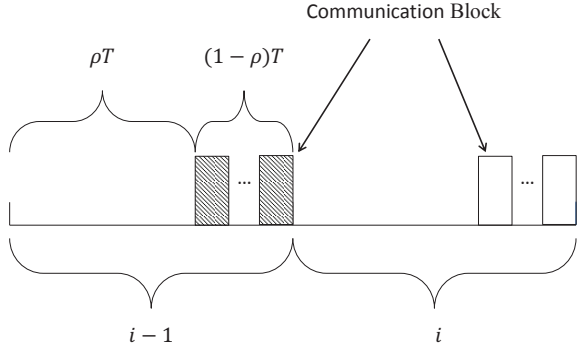}
\caption{Save-Then-Transmit (ST) Protocol } \label{fig:Systemmodel}
\end{figure}

Under the above realistic conditions, we minimize the outage
probability (to be defined in the next section) over $\rho$, when
transmitting over a block fading channel with an arbitrary fading
distribution. In this work, we particularize to the case where the
MESD is a high-efficiency super-capacitor with $\eta = 1$, and the
SESD is a low-efficiency rechargeable battery with $0 \leq \eta \leq 1$.
Based on the outage analysis, we compare the performance between two
system setups: the (new) case with random power supply versus the
(conventional) case with constant power supply, over the Rayleigh
fading channel. It is shown that energy harvesting, which results in
time-varying power availability in addition to the randomness of the
fading channel, may severely degrade the outage performance. To be
concrete, we further consider exponentially
distributed random power, and show that although the diversity order
with exponential power is the same as that with constant power over the
Rayleigh fading channel, the outage probability curve may only display
the slope predicted by this diversity analysis at substantially higher SNRs.

Finally, we extend the ST protocol for the single-channel case
to the general case of wireless network with multiple transmitters.
We propose a time division multiple access (TDMA) based ST (TDMA-ST)
protocol to allocate orthogonal time slots to multiple transmitters
that periodically report to a fusion center. Specifically, we
consider two types of source data at transmitters as follows:
\begin{itemize}
\item {\bf Independent Data}: transmitters send independent data packets to the fusion
center for independent decoding;
\item {\bf Common Data}: transmitters send identical data packets to the fusion center, where diversity combining is applied to decode the common data.
\end{itemize}
It is shown that for both cases if the number of transmitters $N$ is
smaller than the reciprocal of the optimal transmit-ratio ($1-\rho$)
for the single-channel outage minimization, all transmitters can
operate at their individual minimum outage probability. However, as
$N$ goes up and exceeds this threshold, the system-level outage
performance behaves quite differently for the two types of source
data.

The rest of this paper is organized as follows. Section
\ref{sec:system model} presents the system model. Section
\ref{sec:outage minimization} considers finding the optimal
save-ratio for outage minimization and analyzes its various
properties. Section \ref{sec:diversity} compares the outage
performance between fixed power and random power. Section
\ref{sec:multiple transmitter} introduces the TDMA-ST protocol for
the multi-transmitter case. Section \ref{sec:numerical} shows
numerical results. Finally, Section \ref{sec:conclusion} concludes
the paper.

\section{System Model}\label{sec:system model}
\subsection{Definitions and Assumptions}\label{sec:definition assumption}

The block diagram of the system is given in Fig.
\ref{fig:Circuitmodel}. The energy harvested from the
environment\footnote{Wind, solar, geothermal, etc.} is first stored in
either the MESD or the SESD at any given time, as indicated by
switch $a$, before it is used in data transmission. The MESD powers the transmitter directly and usually has
high power density, good recycle ability and high efficiency, e.g. a
super-capacitor \cite{Jayalakshmi2008}. Since the MESD cannot charge
and discharge simultaneously, a SESD (e.g. rechargeable battery)
stores up harvested energy while the transmitter is on, and
transfers all its stored energy to the MESD once the transmitter is
off. We assume in the rest of this paper that the SESD is a battery
with an efficiency $\eta$,\footnote{In practice, the battery efficiency can vary from $60\%$ to $99\%$, depending on different recharging technologies \cite{battery}.} where $\eta \in [0, 1]$. This means that
a fraction $\eta$ of the energy transferred into the SESD during
charging can be subsequently recovered during discharging. The other $1-\eta$ fraction of the energy is thus lost, due to e.g. battery leakage and/or circuit on/off overhead. The reason of choosing a single-throw switch (switch $a$ in Fig. \ref{fig:Circuitmodel}) between the energy harvesting device (EHD) and ESDs is that splitting the harvested energy with a portion going to the SESD, when the transmitter does not draw energy from the MESD, is not energy efficient due to the SESD's lower efficiency. Note that at the current stage of research, the optimal detailed structure of an energy harvesting transmitter is not completely known and there exist various models in the literature (see e.g. \cite{UlukusYener11, Yener12, HoZhang11}). The proposed circuit model, given in Fig. \ref{fig:Circuitmodel}, provides one possible practical design.

\begin{figure}
\centering
\epsfxsize=0.7\linewidth
    \includegraphics[width=8cm]{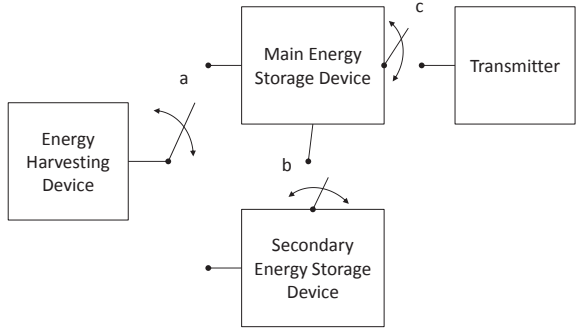}
\caption{Energy Harvesting Circuit Model} \label{fig:Circuitmodel}
\end{figure}

We assume that $Q$ bits of data are generated and must be
transmitted within a time slot of duration $T$ seconds (i.e., delay
constrained). In the proposed ST protocol, the save-ratio $\rho$ is
the reserved fraction of time for energy harvesting by the MESD
within one transmission slot. In other words, data delivery only
takes place in the last $(1-\rho)T$ seconds of each time slot, which
results in an effective rate of $R_{\mathrm{eff}} =
\frac{Q}{(1-\rho)T}$ bits/sec. We also allow for a constant power
consumption of $P_c$ Watts by the transmitter hardware whenever it
is powered on. The combined influence of $\rho$, $\eta$ and $P_c$ on
outage probability is quantified in this work.

Assume a block-fading frequency-nonselective channel, where the channel is constant over the time slot $T$. Over any time slot, the baseband-equivalent channel output is given by
\begin{equation} \label{eq:single channel model}
y = h\cdot x + n,
\end{equation}
where $x$ is the transmitted signal, $y$ is the received signal, and $n$ is i.i.d. circularly symmetric complex Gaussian (CSCG) noise with zero mean and variance $\sigma_n^2$.

For any frame, the ST protocol (cf. Fig. \ref{fig:Systemmodel}) is
described as follows:
\begin{itemize}
\item During time interval $(0, \rho T]$, harvested energy accumulates in
the MESD, which corresponds to the
situation that switches $b$, $c$ are open and $a$ connects to the MESD
in Fig. \ref{fig:Circuitmodel};

\item From time $\rho T$ to $T$, the transmitter is powered on for transmission with energy from the MESD.
Since the transmitter has no knowledge of the channel state, we assume that all the buffered energy in the MESD is used up (best-effort
transmission) in each frame. Since the MESD cannot charge and discharge at the same
time, the SESD starts to store up harvested energy while the
transmitter is on. Referring to Fig. \ref{fig:Circuitmodel}, $c$ is
closed, $b$ is open and $a$ switches to the SESD;

\item At time $T$, the transmitter completes the transmission and powers off. The SESD transfers all its buffered energy to
the MESD within a negligible charging time, at efficiency $\eta$. In
other words, $b$ is closed and switches $a$ and $c$ are open in Fig.
\ref{fig:Circuitmodel}.
\end{itemize}

\subsection{Outage Probability}\label{sec:outage probability definition}
It is clear that the energy harvesting rate $X$ is a non-negative random variable with finite
support, i.e., \ $0 \leq X \leq P_H < \infty$, as the maximum amount
of power one can extract from any source is finite. Suppose $f_X(x)$
and $F_X(x)$ represent its probability density function (PDF) and
cumulative distribution function (CDF), respectively. According to
the proposed ST protocol, the total buffered energy in the MESD at
$t = \rho T$ (the start of data transmission within a transmission
slot) is given by
\begin{equation}\label{eq:buffered energy}
E_T = X\left[\rho+\eta(1-\rho)\right]T.
\end{equation}
Denote $P = \frac{E_T}{(1-\rho)T} =
X\left[\frac{\rho}{1-\rho}+\eta\right]$ as the average total
power, which is constant over the entire transmission period, and
$P_c$ as the circuit power (i.e.\ the power consumed by the hardware
during data transmission), again assumed constant. The mutual
information of the channel (\ref{eq:single channel model})
conditioned on $X$ and the channel gain $h$ is (assuming $P > P_c$)
\begin{equation}\label{eq:mutual information}
R_T = \log_2\left(1+\frac{(P-P_c)|h|^2}{\sigma_n^2}\right) = \log_2\left(1+(P-P_c)\Gamma\right)
\end{equation}
where $\Gamma = \frac{|h|^2}{\sigma_n^2}$ with PDF $f_{\Gamma}(\cdot)$ and CDF $F_{\Gamma}(\cdot)$.

For a transmitter with energy harvesting capability and working
under the ST protocol, the outage event is the union of two mutually
exclusive events: Circuit Outage and Channel Outage. Circuit outage
occurs when the MESD has insufficient energy stored up at $t = \rho
T$ to even power on the hardware for the duration of transmission
i.e.\ $E_T < P_c (1-\rho)T$ or equivalent $P < P_c$. Channel outage
is defined as the MESD having sufficient stored energy but the
channel realization does not support the required target rate
$R_{\mathrm{eff}} = \frac{Q}{(1-\rho)T}$ bits/s.

Recalling that $X \in [0,P_H]$, the probabilities of Circuit Outage and Channel Outage are therefore:
\begin{align}
P_{out}^{circuit} & = \mbox{Pr}\left\{P < P_c\right\} \nonumber \\
& = \left\{ \begin{array}{cl} \displaystyle
F_X\left[\phi(\cdot)\right] & \mbox{if } P_H > \phi(\cdot) \\
1 & \mbox{otherwise}. \end{array}\right. \\
P_{out}^{channel} & = \mbox{Pr}\left\{\log_2\left(1+(P-P_c)\Gamma\right) < R_{\mathrm{eff}}, P > P_c \right\} \nonumber \\
& = \mbox{Pr}\left\{\Gamma < \frac{2^{R_{\mathrm{eff}}}-1}{P-P_c}, P > P_c \right\} \nonumber \\
& = \left\{ \begin{array}{cl} \displaystyle
\int_{\phi(\cdot)}^{P_H}f_X(x)F_{\Gamma}\left[g(\cdot)\right]dx & \mbox{if } P_H > \phi(\cdot) \\
0 & \mbox{otherwise}. \end{array}\right. \label{eq:channel outage}
\end{align}
where $g(\rho, \eta, P_c) =
\frac{2^{\frac{Q}{(1-\rho)T}}-1}{x[\frac{\rho}{1-\rho}+\eta]-P_c}$
and $\phi(\rho, \eta, P_c) = \frac{P_c}{\frac{\rho}{1-\rho} +
\eta}$. Since Circuit Outage and Channel Outage are mutually
exclusive, it follows that
\begin{align}\label{eq:outage definition}
P_{out} & =  P_{out}^{circuit} + P_{out}^{channel} \nonumber \\
& = \left\{ \begin{array}{cl} \displaystyle
F_X\left[\phi(\cdot)\right] +\\
\int_{\phi(\cdot)}^{P_H}f_X(x)F_{\Gamma}\left[g(\cdot)\right]dx & \mbox{if } P_H > \phi(\cdot) \\
1 & \mbox{otherwise}. \end{array}\right.
\end{align}
For convenience, we define
\begin{align}\label{eq:simplified outage}
\hat{P}_{out}(\rho, \eta, P_c) = F_X\left[\phi(\cdot)\right] + \int_{\phi(\cdot)}^{P_H}f_X(x)F_{\Gamma}\left[g(\cdot)\right]dx
\end{align}
where $\hat{P}_{out}(\rho, \eta, P_c) < 1$ and $P_H > \phi(\cdot)$.

Unlike the conventional definition of outage probability in a block
fading channel, which is dependent only on the fading distribution
and a {\em fixed} average transmit power constraint, in an energy
harvesting system with block fading and the ST protocol, {\em both}
transmit power and channel are random, and the resulting outage is
thus a function of the save-ratio $\rho$, the battery efficiency
$\eta$ and the circuit power $P_c$.

\section{Outage Minimization}\label{sec:outage minimization}
In this section, we design the save-ratio $\rho$ for the ST protocol
by solving the optimization problem
\begin{align}
\mathrm{(P1)}:~\mathop{\mathtt{min.}}_{0 \leq \rho \leq 1} &
~~ P_{out} \nonumber
\end{align}
i.e.\ minimize average outage performance $P_{out}$ in
(\ref{eq:outage definition}) over $\rho$, for any given $\eta \in
[0, 1]$ and $P_c \in [0, \infty)$. Denote the optimal (minimum)
outage probability as $P_{out}^*(\eta, P_c)$ and the optimal
save-ratio as $\rho^*(\eta, P_c)$. Note that $\rho \nearrow 1$
represents transmission of a very short burst at the end of each
frame, and the rest of each frame is reserved for MESD energy harvesting.
$\rho =0$ is another special case, in which the energy consumed in
frame $i$ was collected (by the SESD) entirely in frame $i-1$. (P1)
can always be solved through numerical search, but it is challenging
to give a closed-form solution for $\rho^*(\eta, P_c)$ in terms of
$P_c$ and $\eta$ in general. We will instead analyze how
$\rho^*(\eta, P_c)$ varies with $P_c$ and $\eta$ and thereby get
some insights in the rest of this section.

\begin{proposition}\label{proposition:1}
$P_{out}(\rho, \eta, P_c)$ in (\ref{eq:outage definition}) is a
non-increasing function of battery efficiency $\eta$ and a
non-decreasing function of circuit power $P_c$ for $\rho \in [0,
1)$. The optimal value of (P1) $P_{out}^*(\eta, P_c)$ is strictly
decreasing with $\eta$ and strictly increasing with $P_c$.
\end{proposition}
\begin{proof}
Please refer to Appendix \ref{appendix:proof 0}.
\end{proof}

The intuition of Proposition $\ref{proposition:1}$ is clear: If
$\eta$ grows, the energy available to the transmitter can only grow
or remain the same, whatever the values of $\rho$ and $P_c$, hence
$P_{out}$ must be non-increasing with $\eta$; if $P_c$ grows, the
energy available for transmission decreases, leading to higher
$P_{out}$.

\subsection{Ideal System: $\eta = 1$ and $P_c = 0$}\label{sec:idealsystem}
Suppose that circuit power is negligible, i.e.\ all the energy is
spent on transmission, and the SESD has perfect energy-transfer
efficiency. The condition $P_H > P_c/(\frac{\rho}{1-\rho} + \eta)$
is always satisfied, and problem (P1) is simplified to
\begin{align}
\mathrm{(P2)}:~\mathop{\mathtt{min.}}_{0 \leq \rho \leq 1} &
~~ \int_0^{P_H}f_X(x)F_{\Gamma}\left[\frac{(2^{\frac{Q}{(1-\rho)T}}-1)(1-\rho)}{x}\right]dx \nonumber
\end{align}
where the optimal value of (P2) is denoted as $P_{out}^*(1,0)$, and the optimal save-ratio is denoted as $\rho^*(1,0)$.

\begin{lemma}\label{lemma:1}
The minimum outage probability when $\eta = 1$ and $P_c = 0$ is given by
\begin{equation}\label{eq:ideal minimum outage}
P_{out}^*(1,0) = \int_0^{P_H}f_X(x)F_{\Gamma}\left[\frac{2^{Q/T}-1}{x}\right]dx
\end{equation}
and is achieved with the save-ratio $\rho^*(1,0) = 0$.
\end{lemma}

\begin{proof}
Please refer to Appendix \ref{appendix:proof 1}.
\end{proof}

Lemma \ref{lemma:1} indicates that the optimal strategy for a
transmitter that uses no power to operate its circuitry powered by
two ESDs with 100 percent efficiency, is to transmit
continuously.\footnote{Except for the time needed in each slot to
transfer energy from the SESD to the MESD, which we assume to be
negligible.} This is not surprising because the SESD collects energy
from the environment just as efficiently as the MESD does, and so
idling the transmitter while the MESD harvests energy wastes
transmission resources (time) while not reaping any gains (energy
harvested). However, we will see that this is only true when there
is no circuit power and the battery efficiency is perfect.

\subsection{Inefficient Battery: $\eta < 1$ and $P_c = 0$} \label{sec:inefficient battery}
When the SESD energy transfer efficiency $\eta < 1$ and $P_c = 0$, (P1) becomes
\begin{align}
\mathrm{(P3)}:~\mathop{\mathtt{min.}}_{0 \leq \rho \leq 1} &
~~ \int_0^{P_H}f_X(x)F_{\Gamma}\left[\frac{(2^{\frac{Q}{(1-\rho)T}}-1)}{x(\frac{\rho}{1-\rho}+\eta)}\right]dx \nonumber
\end{align}
where the optimal value of (P3) is denoted as $P_{out}^*(\eta,0)$, and the optimal save-ratio is denoted as $\rho^*(\eta,0)$.

\begin{lemma}\label{lemma:2}
When SESD energy transfer efficiency $\eta < 1$ and circuit power $P_c = 0$, the optimal save-ratio $\rho$ has the following properties.
\begin{enumerate}
\item A ``phase transition'' behavior:
\begin{align}\label{eq:phase transition}
\left\{\begin{array}{ll}
\rho^*(\eta, 0) = 0, ~~~~ & \eta \in \left[\frac{2^{\frac{Q}{T}}-1}{2^{\frac{Q}{T}}(\ln2)\frac{Q}{T}}, 1 \right) \\
\rho^*(\eta, 0) > 0, ~~~~ & \eta \in \left[0, \frac{2^{\frac{Q}{T}}-1}{2^{\frac{Q}{T}}(\ln2)\frac{Q}{T}} \right)
\end{array}
\right.
\end{align}
\item $\rho^*(\eta,0)$ is a non-increasing function of $\eta$, for $0 \leq \eta \leq 1$.
\end{enumerate}
\end{lemma}

\begin{proof}
Please refer to Appendix \ref{appendix:proof 2}.
\end{proof}

According to (\ref{eq:phase transition}), if the SESD efficiency is
above a threshold, the increased energy available to the transmitter
if the MESD rather than the SESD collects energy over $[0,\rho T]$
is not sufficient to overcome the extra energy required to transmit
at the higher rate $R_{\mathrm{eff}}$ over $(\rho T, T]$. The result
is that the optimal $\rho$ is 0. On the other hand, if $\eta$ is
below that threshold, then some amount of time should be spent
harvesting energy using the higher-efficiency MESD even at the
expense of losing transmission time. Lemma \ref{lemma:2} quantifies
precisely the interplay among $\eta$, $Q$, $T$ and $\rho$.

We should note here that even though we consider the case of having
two ESD's, by setting $\eta = 0$, we effectively remove the SESD and
hence our analysis applies also to the single-ESD case. According to
(\ref{eq:phase transition}), if we only have one ESD, the optimal
save-ratio is $\rho^*(0,0)$, which is always larger than $0$. This
is intuitively sensible, because with only one ESD obeying the
energy half-duplex constraint, it would be impossible to transmit
all the time ($\rho = 0$) because that would leave no time at all
for energy harvesting.

\subsection{Non-Zero Circuit Power: $\eta \leq 1$, $P_c > 0$}\label{sec:finitepc}
Non-zero circuit power $P_c$ leads to two mutually exclusive
effects: (i) inability to power on the transmitter for the
$(1-\rho)T$ duration of transmission -- this is when $P_H <
\phi(\cdot)$ in (\ref{eq:outage definition}); and (ii) higher outage
probability if $P_H > \phi(\cdot)$ because some power is devoted to
running the hardware.

Since $\frac{P_c}{\frac{\rho}{1-\rho}+\eta}$ decreases as $\rho$
increases, its maximum value is $\frac{P_c}{\eta}$. Therefore, if
$P_H > \frac{P_c}{\eta}$, the transmitter would be able to recover
enough energy (with non-zero probability) to power on the
transmitter, i.e.\ $\rho \in [0, 1)$. If $P_H \leq
\frac{P_c}{\eta}$, by condition $P_H \leq
\frac{P_c}{\frac{\rho}{1-\rho}+\eta}$, save-ratio $\rho$ is required
to be larger than
$\frac{\frac{P_c}{P_H}-\eta}{1-\eta+\frac{P_c}{P_H}}$. In summary,
\begin{itemize}
\item If $P_c < P_H\eta$
\begin{equation}
P_{out} = \hat{P}_{out}(\rho, \eta, P_c), ~~~~ \forall \rho \in [0, 1) \nonumber
\end{equation}

\item If $P_c \geq P_H\eta$
\begin{align}\label{eq:outage seperation rho}
P_{out} = \left\{\begin{array}{ll}
1, ~~~~ & \rho \leq \frac{\frac{P_c}{P_H}-\eta}{1-\eta+\frac{P_c}{P_H}} \\
\hat{P}_{out}(\rho, \eta, P_c), ~~~~ & \rho > \frac{\frac{P_c}{P_H}-\eta}{1-\eta+\frac{P_c}{P_H}}
\end{array}
\right.
\end{align}
\end{itemize}

If $P_c \geq \eta P_H$, referring to (\ref{eq:outage seperation
rho}), we may conclude that $\rho^*(\eta, P_c) >
\frac{\frac{P_c}{P_H}-\eta}{1-\eta+\frac{P_c}{P_H}}$ due to the need to offset circuit power consumption. If $P_c < \eta P_H$,
theoretically, the transmitter is able to recover enough energy
(with non-zero probability for all $\rho \in [0, 1)$) to transmit.

\begin{lemma}\label{lemma:3}
For an energy harvesting transmitter with battery efficiency $\eta$ and non-zero circuit power $P_c$,
\begin{equation}\label{eq:circuit power battery inefficient}
 \eta - \frac{P_c}{P_H} < \frac{2^{\frac{Q}{T}} - 1}{2^{\frac{Q}{T}}(\ln2)\frac{Q}{T}} \;\;\Longrightarrow\;\;\rho^*(\eta, P_c) > 0.
\end{equation}
\end{lemma}
\begin{proof}
Please refer to Appendix \ref{appendix:proof 3}.
\end{proof}

Intuitively, the smaller the circuit power, the more energy we can
spend on transmission; the larger the battery efficiency is, the
more energy we can recover from energy harvesting. Small circuit
power and high battery efficiency suggest continuous transmission
($\rho^*(\eta, P_c) = 0$), which is consistent with our intuition.
According to Lemma \ref{lemma:3}, larger circuit power may be
compensated by larger ESD efficiency (when the threshold for $\eta$
is smaller than 1). A non-zero save-ratio is only desired if there
exists significant circuit power to be offset or substantial ESD
inefficiency to be compensated. The threshold depends on required
transmission rate.

\begin{remark}
It is worth noticing that if  we ignore the battery inefficiency or set $\eta = 1$, Lemma \ref{lemma:3} could be simplified as
\begin{align}\label{eq:circuit power}
P_c > \frac{2^{\frac{Q}{T}}(\ln2)\frac{Q}{T} - 2^{\frac{Q}{T}} +
1}{2^{\frac{Q}{T}}(\ln2)\frac{Q}{T}}P_H \;\;\Longrightarrow\;\;
\rho^*(1, P_c) > 0
\end{align}
where only circuit power $P_c$ impacts the save-ratio.
Since the MESD and the SESD are equivalent ($\eta=1$), harvesting energy using the MESD is not
the reason for delaying transmission. Instead, $\rho^* > 0$ when $P_c$ is so large that we should transmit over a shorter interval at a higher power, so that the actual transmission power minimizes $P_{out}$. Circuit power similarly determined the fundamental
tradeoff between energy efficiency and spectral efficiency (data
rate) in \cite{miao2010}, in which it was shown that with additional
circuit power making use of all available time for transmission is
not the best strategy in terms of both energy and spectral
efficiency. In this paper, outage is minimized through utilizing
available (random) energy efficiently, wherein circuit power causes
a similar effect.
\end{remark}

\section{Diversity Analysis}\label{sec:diversity}
The outage performance of wireless transmission over fading channels
at high SNR can be conveniently characterized by the so-called {\it
diversity order} \cite{Tse}, which is the high-SNR slope of the
outage probability determined from a SNR-outage plot in the log-log
scale. Mathematically, the diversity order is defined as
\begin{align}\label{eq:diversity_definition}
d = - \lim_{\bar{\gamma} \rightarrow \infty} \frac{\log_{10}(P_{out})}{\log_{10}(\bar{\gamma})}
\end{align}
where $P_{out}$ is the outage probability and $\bar{\gamma}$ is the
average SNR.

Diversity order under various fading channel conditions has been
comprehensively analyzed in the literature (see e.g. \cite{Tse} and
references therein). Generally speaking, if the fading channel power
distribution has an accumulated density near zero that can be
approximated by a polynomial term, i.e., $\mbox{Pr}\left(|h|^2 \leq
\epsilon \right) \approx \epsilon^k$, where $\epsilon$ is an
arbitrary small positive constant, then the constant $k$ indicates
the diversity order of the fading channel. For example, in the case
of Rayleigh fading channel with $\mbox{Pr}\left(|h|^2 \leq \epsilon
\right)\approx \epsilon$, the diversity order is thus 1 according to
(\ref{eq:diversity_definition}).

However, the above diversity analysis is only applicable to
conventional wireless systems in which the transmitter has a
constant power supply. Since energy harvesting results in random
power availability in addition to fading channels, the PDF of the
receiver SNR due to both random transmit power and random channel
power may not necessarily be polynomially smooth at the origin (as
we will show later). As a result, the conventional diversity
analysis with constant transmit power cannot be directly applied. In
this section, we will investigate the effect of random power on
diversity analysis, as compared with the
conventional constant-power case. For clarity, in the rest of this
section, we consider the ideal system with $\eta = 1$ and $P_c = 0$,
and the Rayleigh fading channel with $\mathbb{E}[\Gamma] =
\mathbb{E}[\frac{|h|^2}{\sigma_n^2}] =
\frac{\sigma_{h}^2}{\sigma_{n}^2} = \lambda_{\gamma}$.

From (\ref{eq:channel outage}) and (\ref{eq:outage definition}), the outage probability
when $\eta = 1$ and $P_c = 0$ is given by
\begin{align}
P_{out} = \mbox{Pr}\left\{\log_2(1+P\Gamma) < \frac{Q}{(1-\rho)T} \right\}.
\end{align}
Based on Lemma \ref{lemma:1}, the minimum outage
probability is achieved with the save-ratio $\rho = 0$. Therefore,
the outage probability is simplified as\footnote{For convenience,
$P_{out}^{*}$ is used to represent $P_{out}^*(1,0)$ in the rest of
this section.}
\begin{align}\label{eq:outage for diversity analysis}
P_{out}^{*} = \mbox{Pr}\left\{P\Gamma < C\right\} = \int_0^{\infty}\int_0^{\frac{C}{P}}f_P(p)f_{\Gamma}(\gamma)d\gamma dp
\end{align}
where $C = 2^{\frac{Q}{T}} - 1$ and the last equality comes from the
assumption of Rayleigh fading channel so the $\Gamma$ is exponential distributed. It is worth noting that in
this case with $\eta = 1$ and $\rho = 0$, according to
(\ref{eq:buffered energy}), the energy arrival rate $X$ and the
average total power $P$ are identical and thus can be used
interchangeably.

Clearly, the near-zero behavior of $P_{out}^{*}$ critically depends
on the PDF of random power $f_P(p)$, while intuitively we should
expect that random power can only degrade the outage performance. We choose to use the Gamma
distribution to model the random power $P$, because the Gamma distribution models
many positive random variables (RVs)
\cite{Springer, Zhang2008}. The Gamma distribution is very general,
including exponential, Rayleigh, and Chi-Square as special cases;
furthermore, the PDF of any positive continuous RV can be properly
approximated by the sum of Gamma PDFs. Supposing that $P$ follows a Gamma distribution denoted by $P
\sim \mathcal{G}(\beta, \lambda_p)$, then its PDF is given by
\begin{align}\label{eq:Gamma distribution}
f_P(p) = \frac{p^{\beta - 1}\mbox{exp}\left(-\frac{p}{\lambda_p}\right)}{\lambda_p^{\beta}\Gamma(\beta)}U(p)
\end{align}
where $U(\cdot)$ is the unit step function, $\Gamma(\cdot)$ is the
Gamma function, and $\beta>0$, $\lambda_p>0$ are given parameters.
Referring to \cite[Lemma 2]{Nadarajah}, which gives the distribution of the product of $m$ Gamma RVs, the outage probability in
(\ref{eq:outage for diversity analysis}) can be computed as
\begin{align}\label{eq:outage meijerG}
P_{out}^{*} = \frac{1}{\Gamma(\beta)}G_{13}^{21}\left(\frac{C\lambda_{\gamma}}{\lambda_p}\left|\begin{array}{ll}
1 \\
1, \beta, 0
\end{array}\right.
\right)
\end{align}
where $G(\cdot)$ is the Meijer G-function \cite{Springer}.

Meijer G-function can in general only be numerically evaluated and does not give much insights about
how random power affects the outage performance. Next, we further
assume that the random power $P$ is exponentially distributed (as a
special case of Gamma distribution with $\beta = 1$) to demonstrate
the effect of random power.
\begin{lemma}\label{lemma:4}
Suppose that $P$ is exponentially distributed with mean $\lambda_p$,
the channel is Rayleigh fading with $\mathbb{E}[\Gamma] =
\frac{\sigma_{h}^2}{\sigma_{n}^2} = \lambda_{\gamma}$, and thus the
average received SNR $\bar{\gamma} = \lambda_p\lambda_{\gamma} =
\frac{\lambda_p\sigma_h^2}{\sigma_n^2}$. The minimum outage
probability $P_{out}^{*}$, under an ideal system with $\eta = 1$ and
$P_c = 0$, is given by
\begin{align}\label{eq:outage seriel representation}
P_{out}^{*} = \sum_{k=0}^{\infty}\frac{C^{k+1}}{(k!)^2(k+1)\bar{\gamma}^{k+1}}\left[\frac{1}{k+1} - \ln \frac{C}{\bar{\gamma}} + 2\psi(k+1)\right]
\end{align}
where $\psi(x)$ is the digamma function \cite{Abramowitz} and
$\ln\left(\cdot\right)$ represents the natural logarithm.
\end{lemma}
\begin{proof}
Please refer to Appendix \ref{appendix:proof 4}.
\end{proof}

In the asymptotically high-SNR\footnote{We assume that high SNR is achieved via decreasing noise power $\sigma_{n}^2$, while fixing the average harvested energy.} regime, we can approximate
$P_{out}^{*}$ by taking only the first term of (\ref{eq:outage
seriel representation}) as
\begin{align}\label{eq:outage appximation}
P_{out}^{*} & \approx \frac{C}{\bar{\gamma}}\left(1 - \ln{\frac{C}{\bar{\gamma}}} + 2\psi(1)\right) \approx \frac{\ln{\bar{\gamma}}}{\bar{\gamma}}.
\end{align}
As observed, $P_{out}^{*}$ decays as
$\bar{\gamma}^{-1}\ln{(\bar{\gamma})}$ rather than
$\bar{\gamma}^{-1}$ as in the conventional case with constant power,
which indicates that the PDF of the receiver SNR is no longer
polynomially smooth near the origin. Hence, the slope of
$P_{out}^{*}$ in the SNR-outage plot, or the diversity order, will
converge much more slowly to $\bar{\gamma}^{-1}$ with SNR than in the constant-power case,
suggesting that random energy arrival has a significant impact on
the diversity performance. More specifically, we obtain the
diversity order in the case of exponentially distributed random
power as
\begin{align} d  = -
\lim_{\bar{\gamma} \rightarrow \infty} \frac{-
\log_{10}{\bar{\gamma}} +
\log_{10}{\left(\ln{\bar{\gamma}}\right)}}{\log_{10}{\bar{\gamma}}}
= 1
\end{align}
which is, in principle, the same as that over the Rayleigh fading
channel with constant power. We thus conclude that diversity order
may not be as meaningful a metric of
evaluating outage performance in the presence of random power, as in the conventional case of constant power.

\section{Multiple Transmitters}\label{sec:multiple transmitter}
In this section, we extend the ST protocol for the single-channel
case to the more practical case of multiple transmitters in a
wireless network, and quantify the system-level outage performance
as a function of the number of transmitters in the network.

\subsection{TDMA-ST}\label{sec:TDMAST}
We consider a wireless network with $N$ transmitters, each of which
needs to transmit $Q$ bits of data within a time frame of duration
$T$ seconds to a common fusion center (FC). It is assumed that each
transmitter is powered by the same energy harvesting circuit model
as shown in Fig. \ref{fig:Circuitmodel}, and transmits over the
baseband-equivalent channel model given in (\ref{eq:single channel
model}). We also assume a homogeneous system setup, in which the
channel gains, energy harvesting rates or additive noises for all
transmitter-FC links are independent and identically distributed (i.i.d).

\begin{figure}
\centering
\epsfxsize=0.7\linewidth
    \includegraphics[width=8cm]{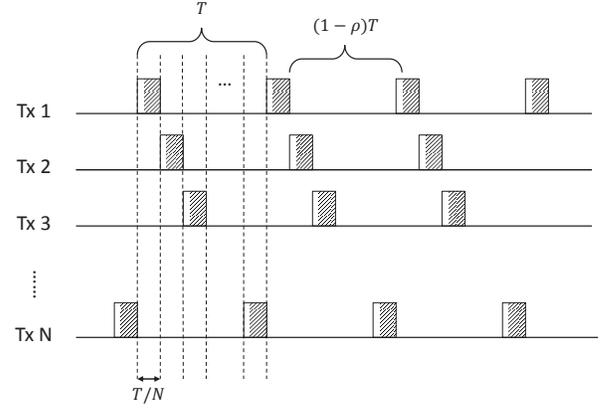}
\caption{TDMA based ST (TDMA-ST)} \label{fig:Multiuser}
\end{figure}

In order to allow multiple transmitters to communicate with
the FC, we propose a TDMA based ST (TDMA-ST) protocol as follows
(cf. Fig. \ref{fig:Multiuser}):
\begin{itemize}
\item Every frame is equally divided into $N$ orthogonal time slots with each slot equal to $\frac{T}{N}$ seconds.

\item Assuming perfect time synchronization, each transmitter is assigned a different (periodically repeating) time slot for transmission, i.e.,
in each frame, transmitter $i$ is allocated the time slot
$\left[\frac{(i-1)}{N}T, ~ \frac{i}{N}T\right)$, $1 \leq i \leq N$.

\item Assuming $\rho_i = \rho$ for all $i$'s, each transmitter
implements the ST protocol with the transmission time in each frame
aligned to be within its assigned time slot; as a result, the
maximum {\it transmit-ratio}, denoted by $1-\rho$, for each
transmitter cannot exceed $1/N$, which means that $\rho \geq 1-
\frac{1}{N}$.
\end{itemize}

The protocol described above is illustrated in Fig.
\ref{fig:Multiuser}. Unlike the single-channel case where the
transmitter can select any save-ratio $\rho$ in the interval $0\leq
\rho\leq 1$, in the case of TDMA-ST, $\rho$ is further constrained
by $\rho \geq 1- \frac{1}{N}$ to ensure orthogonal transmissions by
all transmitters. Due to this limitation, each transmitter may not
be able to work at its individual minimum outage probability unless
the corresponding optimal save-ratio $\rho^*$ satisfies $\rho^* \geq
1- \frac{1}{N}$ or $N \leq \frac{1}{1-\rho^*}$. In this case, ST protocol naturally extends to TDMA-ST with every transmitter operating at the optimal save-ratio $\rho^*$. However, if $N$ exceeds the threshold $\frac{1}{1-\rho^*}$, transmitters have to deviate from $\rho^*$ to maintain orthogonal transmissions. Next, we evaluate
the performance of TDMA-ST for two types of source data at
transmitters: Independent Data and Common Data.

\subsection{Independent Data}\label{sec:Independent information}
First, consider the case where all transmitters send independent
data packets to the FC in each frame, which are decoded separately
at the FC. Under the symmetric setup, for a given $\rho$, all
transmitters should have the same average outage performance.
Consequently, the system-level outage performance in the case of
independent data can be equivalently measured by that of the
individual transmitter, i.e.,
\begin{align}
P_{out}^s = P_{out}(\rho, \eta, P_c).
\end{align}
We can further investigate the following two cases:
\begin{itemize}
\item $N \leq \frac{1}{1-\rho^*}$ \\
In this case, the additional constraint due to TDMA, $\rho^* \geq 1-
\frac{1}{N}$, is satisfied. Since $P_{out}^s$ is the same as that of
the single-transmitter case, the system is optimized when all
transmitters work at their individual minimum outage with save-ratio
$\rho^*$. Thus, the minimum system outage probability is
$P_{out}^{s*} = P_{out}^*(\eta, P_c)$.

\item $N > \frac{1}{1-\rho^*}$ \\
In this case, the TDMA constraint on $\rho^*$ is violated and thus
we are not able to allocate all transmitters the save-ratio
$\rho^*$, which means that each transmitter has to deviate from its
minimum outage point. Since in this case $\rho^*< 1 -
\frac{1}{N}\leq \rho $, the best strategy for each transmitter is to
choose $\rho = 1 - \frac{1}{N} $. Thus, $P_{out}^{s*} =
P_{out}(1-\frac{1}{N},\eta, P_c)$.

\end{itemize}

To summarize, the optimal strategy for each transmitter in the case
of independent data is given by
\begin{align}\label{eq:multiple strategy}
\rho = \left\{\begin{array}{ll}
\rho^*, ~~~~ & N \leq \frac{1}{1-\rho^*}  \\
1 - \frac{1}{N}, ~~~~ & N > \frac{1}{1-\rho^*}
\end{array}
\right.
\end{align}
which implies that the number of transmitters should be kept below
the reciprocal of the single-channel optimal transmit-ratio;
otherwise, the system outage performance will degrade.

\subsection{Common Data}

Next, consider the case where all transmitters send identical data
packets in each frame to the FC, which applies diversity combining
techniques to decode the common data. For simplicity, we consider
selection combining (SC) at the receiver, but similar
results can be obtained for other diversity combining techniques
\cite{Tse}. With SC, the system outage probability is given by
\cite{Tse}
\begin{align}
P_{out}^s = \left(P_{out}(\rho, \eta, P_c)\right)^{N}.
\end{align}
Similarly to the case of independent data, we can get exactly the
same result for the optimal transmit strategy given in
(\ref{eq:multiple strategy}) for the common-data case, with which
the minimum system outage probability is obtained as
\begin{align}
P_{out}^{s*} = \left\{\begin{array}{ll}
\left(P_{out}^*(\eta, P_c)\right)^{N}, ~~~~ & N \leq \frac{1}{1-\rho^*}  \\
\left(P_{out}(1 - \frac{1}{N}, \eta, P_c)\right)^{N}, ~~~~ & N > \frac{1}{1-\rho^*}
\end{array}
\right.
\end{align}
From the above, it is evident that the system outage probability
initially drops as $N$ increases, provided that $N \leq
\frac{1}{1-\rho^*}$. However, when $N
> \frac{1}{1-\rho^*}$, it is not immediately clear whether the
system outage increases or decreases with $N$, since increasing $N$
improves the SC diversity, but also makes each transmitter deviate
even further from its minimum outage save-ratio according to
(\ref{eq:multiple strategy}).

\section{Numerical Examples}\label{sec:numerical}

In this section, we provide numerical examples to validate our
claims. We assume that the energy harvesting rate $X$
follows a uniform distribution (unless specified otherwise) within
$[0, 100]$ (i.e., $P_H = 100$ J/s), and the channel is Rayleigh
fading with exponentially distributed $\Gamma$ with parameter
$\lambda=0.02$. We also assume the target transmission rate
$R_{\mathrm{req}} = \frac{Q}{T} = 2$ bits/s.\footnote{This is
normalized to a bandwidth of 1 Hz, i.e.\ $R_{\mathrm{req}}$ is the
spectral efficiency in bis/s/Hz.}

\begin{figure}
\centering
\epsfxsize=0.7\linewidth
    \includegraphics[width=9cm]{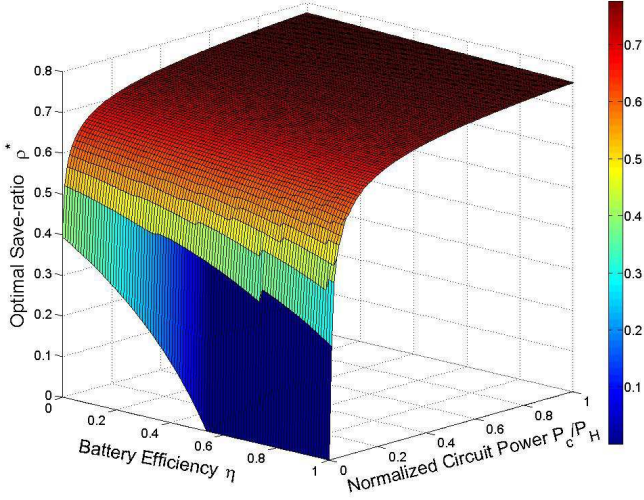}
\caption{Optimal save-ratio $\rho^*$} \label{fig:OptimalRho}
\end{figure}

Fig. \ref{fig:OptimalRho} demonstrates how battery efficiency $\eta$
and circuit power $P_c$ affect the optimal save-ratio $\rho^*$ for
the single-channel case. As observed, larger $P_c$ and smaller
$\eta$ result in larger $\rho^*$, i.e. shorter transmission time.
Since the increment is more substantial along $P_c$ axis, circuit
power has a larger influence on the optimal save-ratio compared with
battery efficiency. $\rho^*(1,0) = 0$ verifies the result of Lemma
\ref{lemma:1} for an ideal system, while $\rho^*(\eta, 0)$ along the
line $P_c = 0$ demonstrates the ``phase transition'' behavior stated
in Lemma \ref{lemma:2}. The transition point is observed to be $\eta
= 0.541$, which can also be computed from (\ref{eq:phase
transition}).

\begin{figure}
\centering
\epsfxsize=0.7\linewidth
    \includegraphics[width=9cm]{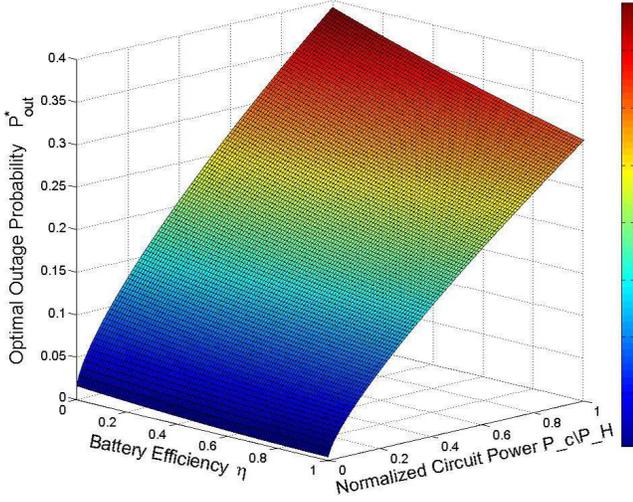}
\caption{Optimal outage probability $P_{out}^*$} \label{fig:OptimalOutage}
\end{figure}

Fig. \ref{fig:OptimalOutage} shows the optimal
(minimum) outage probability $P_{out}^*(\eta, P_c)$ corresponding to $\rho^*$ in
Fig. \ref{fig:OptimalRho}. Consistent with Proposition
\ref{proposition:1}, $P_{out}^*(\eta, P_c)$ is observed to be
monotonically decreasing with battery efficiency $\eta$ and
monotonically increasing with circuit power $P_c$. Again, $P_c$ affects outage performance more significantly
than $\eta$. From Fig. \ref{fig:OptimalOutage}, we see that for a
reasonable outage probability e.g. below $0.05$, $P_c$ has to be
small and $\eta$ has to be close to $1$. Our results can thus be
used to find the feasible region in the $\eta-P_c$ plane for a given
allowable $P_{out}$.

\begin{figure}
\centering
\epsfxsize=0.7\linewidth
    \includegraphics[width=8cm]{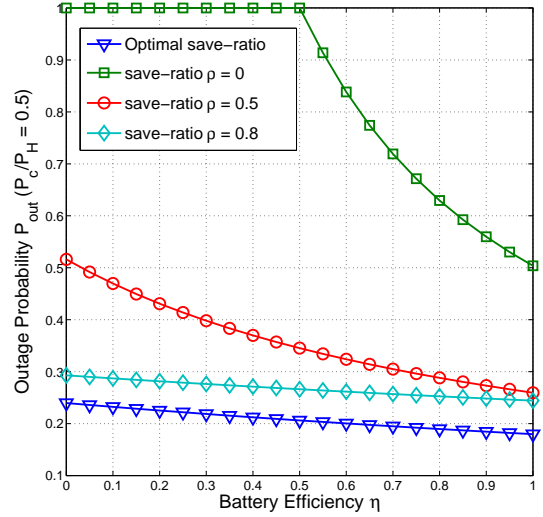}
\caption{Outage performance comparison: $\frac{P_c}{P_H} = 0.5$} \label{fig:CompareEta}
\end{figure}

\begin{figure}
\centering
\epsfxsize=0.7\linewidth
    \includegraphics[width=8cm]{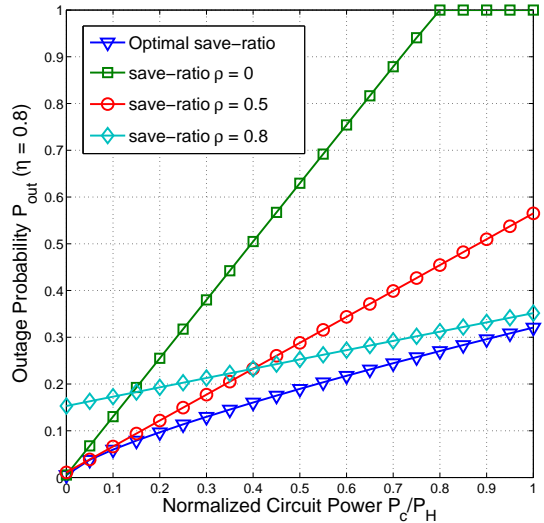}
\caption{Outage performance comparison: $\eta = 0.8$} \label{fig:ComparePc}
\end{figure}

Figs. \ref{fig:CompareEta} and \ref{fig:ComparePc} compare the
outage performance with versus without save-ratio optimization. In
Fig. \ref{fig:CompareEta} we fix the normalized circuit power
$\frac{P_c}{P_H} = 0.5$, while in Fig. \ref{fig:ComparePc} we fix
the battery efficiency $\eta = 0.8$. We observe that optimizing the
save-ratio can significantly improve the outage performance. It is
worth noting that $P_{out}$ has an approximately linear relationship
with the normalized circuit power $\frac{P_c}{P_H}$ as observed in
Fig. \ref{fig:ComparePc}, which indicates that $P_c$ considerably
affects the outage performance as stated previously.

\begin{figure}
\centering
\epsfxsize=0.7\linewidth
    \includegraphics[width=8cm]{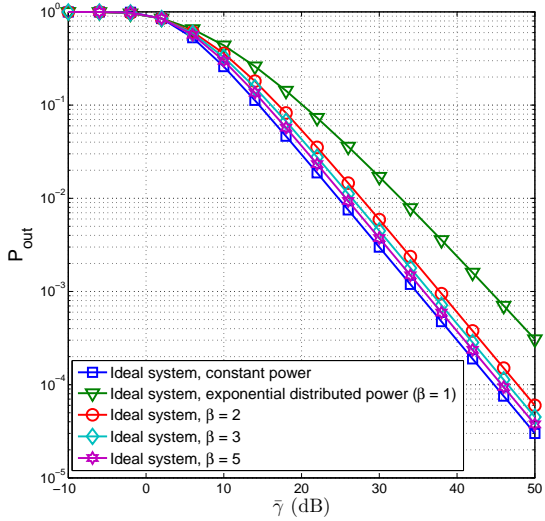}
\caption{Outage probability for an ideal ($\eta = 1$, $P_c = 0$) system with constant power
versus random power} \label{fig:Gammapower}
\end{figure}

\begin{figure}
\centering
\epsfxsize=0.7\linewidth
    \includegraphics[width=8cm]{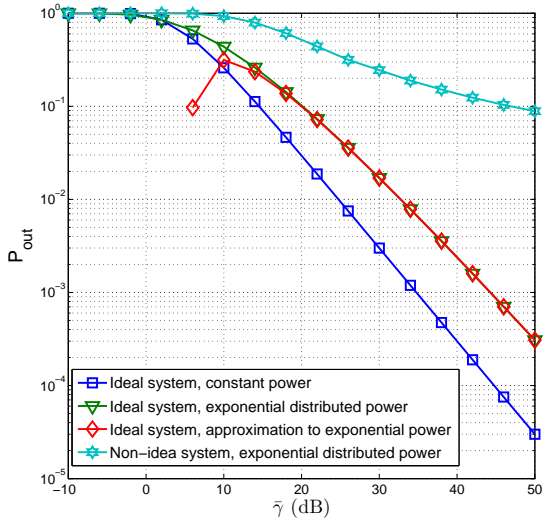}
\caption{Outage probability comparison for ideal ($\eta = 1$, $P_c = 0$) versus non-ideal ($\eta = 0.8$, $P_c = 0.1*\mathbb{E}[P]$)
systems} \label{fig:OutageConverge}
\end{figure}

In Fig. \ref{fig:Gammapower}, the outage probability for an ideal
system ($\eta = 1$, $P_c = 0$) is shown by numerically evaluating (\ref{eq:outage meijerG}).
By fixing the mean value of $P$ as $\mathbb{E}[P] = 50$ J/s and
varying $\beta$ for the Gamma distributed power from $1$ to $5$, the
resulting outage performance is compared with the case of constant
power. As observed, the outage probability increases due to the
existence of power randomness. As $\beta$ increases, the outage
curve approaches the case of constant power. In Fig.
\ref{fig:OutageConverge}, we also plot the outage probability for
the ideal system with exponentially distributed power based on the
approximation given in (\ref{eq:outage appximation}), as well as for
a non-ideal system with the normalized circuit power
$\frac{P_c}{\mathbb{E}[P]} = 0.1$ and battery efficiency $\eta = 0.8$. In
comparison with the constant-power case, for the case of ideal
system we observe that the high-SNR slope or diversity order with
random power clearly converges much slower with SNR, which is in
accordance with our analysis in Section \ref{sec:diversity}.
Furthermore, at $P_{out} = 10^{-3}$, there is about $10$ dB power
penalty observed due to exponential random power, even with the same
diversity order as the constant-power case. It is also observed that
there is a small rising part for the outage approximation given in
(\ref{eq:outage appximation}), since this approximation is only
valid for sufficiently high SNR values ($\bar{\gamma} > 10$ dB). Finally,
 it is worth noting that the outage probability for the
non-ideal system eventually saturates with SNR
(regardless of how small the noise power is or how large the SNR
is), which indicates that the diversity order is zero for any non-ideal system.

\begin{figure}
\centering
\epsfxsize=0.7\linewidth
    \includegraphics[width=8cm]{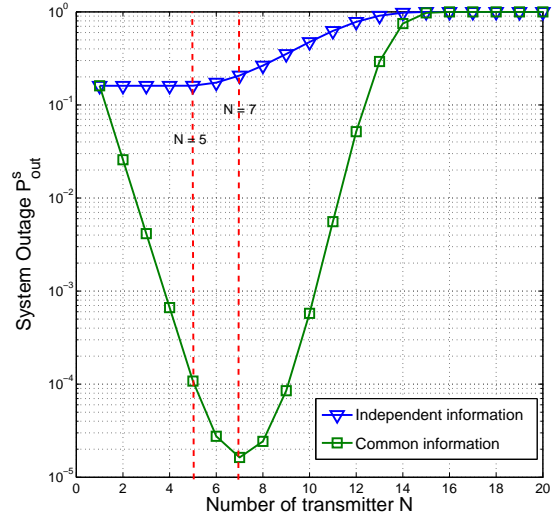}
\caption{Outage performance of multiple transmitters under TDMA-ST
protocol, with $\frac{1}{1 - \rho^*} = 4.83$} \label{fig:MultiOutage}
\end{figure}

Fig. \ref{fig:MultiOutage} shows the outage performance for the case
of multiple transmitters operating under the TDMA-ST protocol. We
set the normalized circuit power $\frac{P_c}{P_H} = 0.5$ and the
battery efficiency $\eta = 0.9$. Then, the optimal save-ratio
$\rho^*$ for single-transmitter outage minimization can be obtained
as $0.7930$ by numerical search. Therefore, the threshold value for
$N$ in the optimal rule of assigning save-ratio values in
(\ref{eq:multiple strategy}) is $\frac{1}{1 - \rho^*} = 4.83$. For
the case of independent data, it is observed that when $N \leq 4$,
the system outage probability is constantly equal to the optimal
single-transmitter outage probability $P_{out}^*(0.9,0.5P_H)$;
however, as $N>4$, the outage probability increases dramatically. In
contrast, for the case of common data, it is observed that the
system outage probability decreases initially as $N$ increases, even
after the threshold value and until $N = 7$, beyond which it starts
increasing. This implies that there is an optimal decision on the
number of transmitters to achieve the optimal outage performance.

\section{Conclusion}\label{sec:conclusion}
In this paper, we studied a wireless system under practical energy
harvesting conditions. Assuming a general model with non-ideal
energy storage efficiency and transmit circuit power, we proposed a
Save-then-Transmit (ST) protocol to optimize the system outage
performance via finding the optimal save-ratio. We characterized how
the optimal save-ratio and the minimum outage probability vary with
practical system parameters. We compared the outage performance
between random power and constant power under the assumption of
Rayleigh fading channel. It is shown that random power considerably
degrades the outage performance. Furthermore, we presented a TDMA-ST
protocol for wireless networks with multiple transmitters. In
particular, two types of source data are examined: independent data
and common data. It is shown that if the number of transmitters is
smaller than the reciprocal of the optimal transmit-ratio for
single-transmitter outage minimization, each transmitter should work
with its minimum outage save-ratio; however, when the number of
transmitters exceeds this threshold, each transmitter has to deviate
from its individual optimal operating point.

There are important problems that remain
unaddressed in this paper and are worth investigating in the future.
For example, we may consider the effect of different
configurations of battery/supercapacitor and MESD/SESD on the system
performance. It is also interesting to investigate the
information-theoretic limits for the ST protocol in the case of
multiple transmitters using more sophisticated multiple-access
techniques other than the simple TDMA.

\appendices

\section{Proof of Proposition \ref{proposition:1}}\label{appendix:proof 0}
According to the Fundamental Theorem of Calculus \cite{Walter1987},
we can derive the first derivative of $\hat{P}_{out}(\rho, \eta,
P_c)$ in (\ref{eq:simplified outage}) with respect to $\eta$, $P_c$
and $\rho$ as
\begin{align}
\frac{\partial \hat{P}_{out}}{\partial\beta} & = \left(\frac{P_c}{\frac{\rho}{1-\rho}+\eta}\right)^{'}f_x\left(\frac{P_c}{\frac{\rho}{1-\rho}+\eta}\right) \nonumber \\
& - \left(\frac{P_c}{\frac{\rho}{1-\rho}+\eta}\right)^{'}f_x\left(\frac{P_c}{\frac{\rho}{1-\rho}+\eta}\right)F_{\Gamma}(\infty) \nonumber \\
& +\int_{\frac{P_c}{\frac{\rho}{1-\rho}+\eta}}^{P_H}f_X(x)f_{\Gamma}\left[g(\cdot)\right]\frac{\partial g(\cdot)}{\partial\beta}dx \nonumber \\
& = \int_{\frac{P_c}{\frac{\rho}{1-\rho}+\eta}}^{P_H}f_X(x)f_{\Gamma}\left[g(\cdot)\right]\frac{\partial g(\cdot)}{\partial\beta}dx
\end{align}
where $\beta$ could be $\eta$, $P_c$ or $\rho$ and $g(\rho, \eta, P_c) = \frac{2^{\frac{Q}{(1-\rho)T}}-1}{x[\frac{\rho}{1-\rho}+\eta]-P_c}$.

It is easy to verify that $\frac{\partial g(\rho, \eta, P_c)}{\partial\eta} < 0, \forall \eta \in [0, 1]$ and $\frac{\partial g(\rho, \eta, P_c)}{\partial P_c} > 0, \forall P_c \in [0, \infty]$. Therefore $\hat{P}_{out}(\rho, \eta, P_c)$ is strictly decreasing with battery efficiency $\eta$ and strictly increasing with circuit power $P_c$. Next, we are going to prove the monotonicity of $P_{out}$ and $P_{out}^*$ with battery efficiency $\eta$, where circuit power $P_c$ is treated as constant.

The condition $P_H > \frac{P_c}{\frac{\rho}{1-\rho}+\eta}$ in (\ref{eq:outage definition}) could be expressed in terms of battery efficiency: $\eta > \frac{P_c}{P_H} - \frac{\rho}{1-\rho}$, then
\begin{align}
P_{out} = \left\{\begin{array}{ll}
1, ~~~~ & \eta \leq \frac{P_c}{P_H} - \frac{\rho}{1-\rho} \\
\hat{P}_{out}, ~~~~ & \eta > \frac{P_c}{P_H} - \frac{\rho}{1-\rho}
\end{array}
\right..
\end{align}

Consider the following two cases:
\begin{itemize}
\item Suppose $\eta_1 < \eta_2$ and $\frac{P_c}{P_H} - \frac{\rho}{1-\rho} > 0$
\begin{itemize}
    \item If $\frac{P_c}{P_H} - \frac{\rho}{1-\rho} < \eta_1 < \eta_2 $, then
     $P_{out}(\rho, \eta_1, P_c) = \hat{P}_{out}(\rho, \eta_1, P_c)$ and $P_{out}(\rho, \eta_2, P_c) = \hat{P}_{out}(\rho, \eta_2, P_c)$. Since $\hat{P}_{out}(\rho, \eta, P_c)$
     is strictly decreasing with battery efficiency $\eta$, we have
    \begin{equation}
    P_{out}(\rho, \eta_1, P_c) > P_{out}(\rho, \eta_2, P_c). \nonumber
    \end{equation}
    \item If $\eta_1 \leq \frac{P_c}{P_H} - \frac{\rho}{1-\rho} < \eta_2$, then $P_{out}(\rho, \eta_1, P_c) = 1$ and $P_{out}(\rho, \eta_2, P_c) = \hat{P}_{out}(\rho, \eta_2, P_c)$. Therefore
    \begin{equation}
    P_{out}(\rho, \eta_1, P_c) = 1 > P_{out}(\rho, \eta_2, P_c). \nonumber
    \end{equation}
    \item If $\eta_1 < \eta_2 \leq \frac{P_c}{P_H} - \frac{\rho}{1-\rho}$, then $P_{out}(\rho, \eta_1, P_c) = P_{out}(\rho, \eta_2, P_c) = 1$, which means
    \begin{equation}
    P_{out}(\rho, \eta_1, P_c) = P_{out}(\rho, \eta_2, P_c). \nonumber
    \end{equation}
\end{itemize}
\item Suppose $\eta_1 < \eta_2$ and $\frac{P_c}{P_H} - \frac{\rho}{1-\rho} \leq 0$, we have $\frac{P_c}{P_H} - \frac{\rho}{1-\rho} \leq \eta_1 < \eta_2$. Then it could be easily verified that
\begin{equation}
P_{out}(\rho, \eta_1, P_c) > P_{out}(\rho, \eta_2, P_c). \nonumber
\end{equation}
\end{itemize}

Combining all the above cases, we can conclude that $P_{out}(\rho,
\eta, P_c)$ is a non-increasing function of battery efficiency
$\eta$ given any non-zero circuit power $P_c$ for $\rho \in [0, 1)$.
Next, we proceed to prove the monotonicity of $P_{out}^*(\eta,
P_c)$.

Assuming $\eta_1 < \eta_2$ again, then we could argue that
$\frac{P_c}{P_H} - \frac{\rho_1^*}{1-\rho_1^*} < \eta_1$ and
$\frac{P_c}{P_H} - \frac{\rho_2^*}{1-\rho_2^*} < \eta_2$, where
$\rho_1^*$ and $\rho_2^*$ are the optimal save-ratio for $\eta =
\eta_1$ and $\eta=\eta_2$, respectively. Therefore we only need to
consider two cases: $\max\left\{\frac{P_c}{P_H} -
\frac{\rho_1^*}{1-\rho_1^*}, \frac{P_c}{P_H} -
\frac{\rho_2^*}{1-\rho_2^*}\right\} < \eta_1 < \eta_2$ and $\eta_1
\leq \max\left\{\frac{P_c}{P_H} - \frac{\rho_1^*}{1-\rho_1^*},
\frac{P_c}{P_H} - \frac{\rho_2^*}{1-\rho_2^*}\right\} < \eta_2$.
From the arguments we have given for the proof of the monotonicity
of $P_{out}$ we know that, under these two conditions we have
\begin{equation}
P_{out}(\rho, \eta_1, P_c) > P_{out}(\rho, \eta_2, P_c). \nonumber
\end{equation}
Therefore,
\begin{equation}
P_{out}^*(\eta_1, P_c) > P_{out}(\rho_1^*, \eta_2, P_c) \geq P_{out}^*(\eta_2, P_c) \nonumber
\end{equation}
which completes the proof of the monotonicity for $P_{out}^*(\eta,
P_c)$. With similar arguments, we could get the results regarding
circuit power $P_c$. Proposition \ref{proposition:1} is thus proved.

\section{Proof of Lemma \ref{lemma:1}}\label{appendix:proof 1}
Since $F_\Gamma(\cdot)$ is non-negative and non-decreasing, we have
\begin{equation}
   a < b \;\;\Rightarrow \;\;F_\Gamma\left(\frac{a}{x}\right) \leq F_\Gamma\left(\frac{b}{x}\right)
\end{equation}
for any $x \in [0,P_H]$. Since $f_X(\cdot)$ is non-negative, this leads to
\begin{align}
   a < b \;\;&\Rightarrow\;\; \int_0^{P_H} f_X(x)F_\Gamma\left(\frac{a}{x}\right) dx \nonumber \\
   &\leq \int_0^{P_H} f_X(x)F_\Gamma\left(\frac{b}{x}\right) dx. \nonumber
\end{align}
Given the form of $P_{out}$ in Problem (P2), with $\rho$ appearing only in the numerator of the argument of $F_\Gamma(\cdot)$, we conclude that $P_{out}$ is a non-decreasing function of $g(\rho) = \left(2^{\frac{Q}{(1-\rho)T}}-1\right)(1-\rho)$. Hence minimizing $g(\rho)$ is equivalent to minimizing $P_{out}$. The first and second derivatives of $g(\rho)$ are
\begin{align}
   g'(\rho) & = 2^{\frac{Q}{(1-\rho)T}}(\ln2)\frac{Q}{(1-\rho)T} - 2^{\frac{Q}{(1-\rho)T}} + 1 \nonumber \\
   g''(\rho) & = 2^{\frac{Q}{(1-\rho)T}}(\ln2)^2\frac{Q^2}{T^2(1-\rho)^3} > 0 \quad\mbox{since $Q > 0$.} \nonumber
\end{align}

Let $h(\rho) = g'(\rho)$. From the second equation above, $h(\rho)$ is an increasing function. In the range $0 \leq \rho \leq 1$, $h(\rho)$ is thus minimized at $\rho = 0$, i.e.\ the minimum of $g'(\rho)$ is $h(0)$, given by
\begin{eqnarray}
   g'_{\mathrm{min}} &=& 2^{\frac{Q}{T}}(\ln 2)\frac{Q}{T} - 2^{\frac{Q}{T}} + 1 \\
   &=& 2^{\frac{Q}{T}}\left(\ln 2^{Q/T} - 1\right) + 1 > 0
\end{eqnarray}
for $Q > 0$. In other words, the smallest value that the gradient of
$g(\rho)$ can take in the range $0 \leq \rho \leq 1$ for any
feasible $Q$ is positive, which implies that $g(\rho)$ is increasing
and therefore minimized at $\rho = 0$, as claimed. The proof of
Lemma \ref{lemma:1} is thus completed.

\section{Proof of Lemma \ref{lemma:2}}\label{appendix:proof 2}
To prove Property 1, we observe that as noted in the proof of Lemma \ref{lemma:1}, $P_{out}$ is a monotonic function of $g(\rho) = \frac{(2^{\frac{Q}{(1-\rho)T}}-1)}{(\frac{\rho}{1-\rho}+\eta)}$ in Problem (P3), hence minimizing $g(\rho)$ leads to the same solution as minimizing $P_{out}$. The first derivative of $g(\rho)$ is
\begin{align}
  g'(\rho) & = \frac{2^{\frac{Q}{(1-\rho)T}}(\ln2)\frac{Q}{(1-\rho)T}[\rho+\eta(1-\rho)] - 2^{\frac{Q}{(1-\rho)T}} + 1}{[\rho+\eta(1-\rho)]^2} \nonumber \\
  & = \frac{u(\rho)}{[\rho+\eta(1-\rho)]^2}. \nonumber
\end{align}
It is clear in the above that the sign of $g'(\rho)$ is the same as
that of $u(\rho)$. Since $u(1) = +\infty$ and $u(\rho)$ is a
differentiable function, if $u(0)$ is negative then there exists a
value $\rho_c \in (0,1)$ such that $u(\rho_c) = 0 = g'(\rho_c)$. It
is easily verified that $u'(\rho) > 0$; hence $\rho_c$ is the
unique optimal value of $\rho$ in this case. Conversely, if there
exists an $\rho_c$ such that $u(\rho_c) = 0$, then $u(0)$ must be
negative. Hence $u(0) < 0$ is a necessary and sufficient condition
for the optimal $\rho$ to lie in $(0,1)$.

The condition $u(0) < 0$ translates into the following condition on $\eta$, which proves the first part of the lemma:
\begin{align} \label{eq:batterycond}
   u(0) < 0 \;\;& \Rightarrow\;\; 2^{\frac{Q}{T}}(\ln2)\frac{Q}{T}\eta - 2^{\frac{Q}{T}} + 1  < 0 \nonumber \\
   & \Rightarrow ~~ \eta < \frac{2^{\frac{Q}{T}}-1}{2^{\frac{Q}{T}}(\ln2)\frac{Q}{T}}.
\end{align}

To prove the second point, suppose  $\rho_1^*(\eta_1,0)$ and
$\rho_2^*(\eta_2,0)$ are optimal save-ratios of (P3) for SESD
efficiencies $\eta_1$ and $\eta_2$, where $\eta_1 < \eta_2$. Then,
$u(\rho_1^*, \eta_1) = 0$ and $u(\rho_2^*, \eta_2) = 0$. Since
$\eta_1 < \eta_2$ and $u(\rho, \eta)$ is an increasing function of
$\eta$, we have $u(\rho_1^*, \eta_2) > 0$. Combining what we have
that $u(\rho, \eta)$ is an increasing function of $\rho$,
$u(\rho_2^*, \eta_2) = 0$ and $u(\rho_1^*, \eta_2) > 0$, we may
conclude $\rho_2^*(\eta_2,0) < \rho_1^*(\eta_1,0)$. Lemma
\ref{lemma:2} is thus proved.

\section{Proof of Lemma \ref{lemma:3}}\label{appendix:proof 3}
According to the proof of Proposition \ref{proposition:1}, the first derivative of $\hat{P}_{out}(\rho, \eta, P_c)$ with respect to $\eta$, $P_c$ and $\rho$ is,
\begin{align}
\frac{\partial \hat{P}_{out}}{\partial\beta}  = \int_{\frac{P_c}{\frac{\rho}{1-\rho}+\eta}}^{P_H}f_X(x)f_{\Gamma}\left[g(\cdot)\right]\frac{\partial g(\cdot)}{\partial\beta}dx \nonumber
\end{align}
where $\beta$ could be $\eta$, $P_c$ or $\rho$ and $g(\rho, \eta,
P_c) =
\frac{2^{\frac{Q}{(1-\rho)T}}-1}{x[\frac{\rho}{1-\rho}+\eta]-P_c}$.
Furthermore, we have
\begin{align}
\frac{\partial g(\rho)}{\partial\rho} & = \frac{2^{\frac{Q}{(1-\rho)T}}(\ln2)\frac{Q}{(1-\rho)T}\left[\rho+\eta(1-\rho) - (1-\rho)\frac{P_c}{x}\right]}{x\left[\rho+\eta(1-\rho) - (1-\rho)\frac{P_c}{x}\right]^{2}} \nonumber \\
& - \frac{2^{\frac{Q}{(1-\rho)T}} - 1}{x\left[\rho+\eta(1-\rho) - (1-\rho)\frac{P_c}{x}\right]^{2}} \nonumber \\
& = \frac{v(\rho)}{x\left[\rho+\eta(1-\rho) - (1-\rho)\frac{P_c}{x}\right]^{2}}. \nonumber
\end{align}
With similar arguments about $u(\rho)$ in the proof of Lemma \ref{lemma:2}, we claim that $v(0) < 0, \forall x \in (\frac{P_c}{\frac{\rho}{1-\rho}+\eta}, P_H]$ is a sufficient condition of having $\rho^*(\eta, P_c) > 0$ while $P_c < \eta P_H$.

Since $v(0)$ is an increasing function of $x$, the condition $v(0) < 0, \forall x \in (\frac{P_c}{\frac{\rho}{1-\rho}+\eta}, P_H]$ translates into the following condition on $\eta$ and $P_c$
\begin{align}
v(0) & = 2^{\frac{Q}{T}}(\ln2)\frac{Q}{T}\left(\eta - \frac{P_c}{x}\right) - 2^{\frac{Q}{T}} + 1 \nonumber \\
& < 2^{\frac{Q}{T}}(\ln2)\frac{Q}{T}\left(\eta - \frac{P_c}{P_H}\right) - 2^{\frac{Q}{T}} + 1 < 0 \nonumber \\
& \Longrightarrow 0 < \eta - \frac{P_c}{P_H} < \frac{2^{\frac{Q}{T}} - 1}{2^{\frac{Q}{T}}(\ln2)\frac{Q}{T}}. \nonumber
\end{align}
Combined with the fact that $\rho^*(\eta, P_c) > \frac{\frac{P_c}{P_H}-\eta}{1-\eta+\frac{P_c}{P_H}}$ when $P_c \geq \eta P_H$, we may conclude
\begin{align}\label{eq:circuit battery region}
\rho^*(\eta, P_c) > 0, ~~~~ \eta - \frac{P_c}{P_H} < \frac{2^{\frac{Q}{T}} - 1}{2^{\frac{Q}{T}}(\ln2)\frac{Q}{T}}.
\end{align}
Lemma \ref{lemma:3} is thus proved.

\section{Proof of Lemma \ref{lemma:4}}\label{appendix:proof 4}
Let $Z = P\Gamma$, where $P$ and $\Gamma$ are exponential random
variables with mean $\lambda_p$ and $\lambda_{\gamma}$ respectively.
Then the PDF of $Z$ could be derived as follows,
\begin{align}
F_Z(z) & = \mbox{Pr}\left\{P\Gamma \leq z \right\} \nonumber \\
& = 1 - \frac{1}{\lambda_p}\int_0^{\infty}e^{-\frac{z}{p\lambda_{\gamma}}}e^{-\frac{p}{\lambda_p}}dp \nonumber \\
& = 1 - 2\sqrt{\frac{z}{\lambda_p\lambda_{\gamma}}}K_1\left(2\sqrt{\frac{z}{\lambda_p\lambda_{\gamma}}}\right)
\end{align}
where $K_1(x)$ is the first-order modified Bessel function of the
second kind and the last equality is given by
\cite[\S3.324.1]{II1980}:
\begin{align}
\int_0^{\infty}\exp\left(-\frac{\beta}{4x} - \gamma x\right)dx = \sqrt{\frac{\beta}{\gamma}}K_1\left(\sqrt{\beta\gamma}\right) \nonumber
\end{align}
where $\Re(\beta) \geq 0, \Re(\gamma) \geq 0$. Let $M =
\frac{1}{\sqrt{\lambda_p\lambda_{\gamma}}}$. Taking the derivative
of $F(z)$ yields
\begin{align}\label{eq:doublepdf}
f(z) & = M\left\{-\frac{1}{\sqrt{z}}K_1\left(2M\sqrt{z}\right) - 2\sqrt{z}\left(K_1\left(2M\sqrt{z}\right)\right)^{'}\right\} \nonumber \\
& = M\left\{-\frac{1}{\sqrt{z}}K_1\left(2M\sqrt{z}\right)\right. \nonumber \\
& - \left.2\sqrt{z}\left(-K_0(2M\sqrt{z})-\frac{1}{2M\sqrt{z}}K_1(2M\sqrt{z})\right)\frac{M}{\sqrt{z}}\right\} \nonumber \\
& = 2M^2K_0\left(2M\sqrt{z}\right) \nonumber \\
& = \frac{2}{\lambda_p\lambda_{\gamma}}K_0\left(2\sqrt{\frac{z}{\lambda_p\lambda_{\gamma}}}\right)
\end{align}
where $\frac{\partial K_v(z)}{\partial z} = - K_{v-1}(z) - \frac{v}{z}K_v(z)$.

Next, we characterize the outage probability using
(\ref{eq:doublepdf}). According to (\ref{eq:outage for diversity
analysis}), we have
\begin{align}
P_{out}^{*} &= \mbox{Pr}\left[P\Gamma < C\right] \nonumber \\
& =
\int_{0}^{C}\frac{2}{\lambda_p\lambda_{\gamma}}K_0\left(2\sqrt{\frac{z}{\lambda_p\lambda_{\gamma}}}\right)dz.
\end{align}
Let $X = \frac{z}{\lambda_p\lambda_{\gamma}}$ and $D =
\frac{C}{\lambda_p\lambda_{\gamma}}$. We then have
\begin{align}\label{eq:outage bessel}
P_{out}^{*} = 2\int_0^{D}K_0\left(2\sqrt{x}\right)dx.
\end{align}
Using the series presentation \cite[\S8.447.3]{II1980}, we have
\begin{align}
K_0(x) = -\ln\left(\frac{x}{2}\right)I_0(x)+ \sum_{k=0}^{\infty}\frac{x^{2k}}{2^{2k}(k!)^2}\psi(k+1)
\end{align}
with the series expansion for the modified Bessel function given by
\begin{align}
I_0(x) = \sum_{k=0}^{\infty}\frac{x^{2k}}{2^{2k}(k!)^2}.
\end{align}
(\ref{eq:outage bessel}) could be expanded as
\begin{align}\label{eq:expandedoutage}
P_{out}^{*} = \sum_{k=0}^{\infty}\frac{2}{(k!)^2}\left[-\frac{1}{2}\int_0^D x^k\ln xdx + \psi(k+1)\int_0^{D}x^kdx\right]
\end{align}
where
\begin{align}
\psi(x) = \frac{d}{dx}\ln \Gamma(x) = \frac{\Gamma(x)^{'}}{\Gamma(x)}
\end{align}
is the digamma function \cite{Abramowitz}. Since the two integrals in (\ref{eq:expandedoutage}) could be evaluated as
\begin{align}
\int_0^Dx^kdx & = \frac{D^{k+1}}{k+1} \nonumber \\
\int_0^D x^k\ln xdx & = \left.x^{k+1}\left(\frac{\ln x}{k+1} - \frac{1}{(k+1)^2}\right)\right|_{x=0}^{x=D} \nonumber \\
& = D^{k+1}\left(\frac{\ln D}{k+1} - \frac{1}{(k+1)^2}\right) \nonumber
\end{align}
where $\lim_{x \rightarrow 0} (x\ln{x}) = 0$. Then we have
\begin{align}
P_{out}^{*} =
\sum_{k=0}^{\infty}\frac{2}{(k!)^2}\frac{D^{k+1}}{k+1}\left[-\frac{1}{2}\left(\ln
D - \frac{1}{k+1}\right) + \psi(k+1)\right].
\end{align}
Since $D = \frac{C}{\lambda_p\lambda_{\gamma}} =
\frac{C\sigma_n^2}{\lambda_p\sigma_h^2} = \frac{C}{\bar{\gamma}}$,
(\ref{eq:outage seriel representation}) follows. Lemma \ref{lemma:4}
is thus proved.

\end{document}